\documentclass[conference,]{IEEEtran}
\IEEEoverridecommandlockouts

\usepackage{cite}
\usepackage{amsmath,amssymb,amsfonts}
\usepackage{algorithmic}
\usepackage{graphicx}
\usepackage{textcomp}
\usepackage{xcolor}
\def\BibTeX{{\rm B\kern-.05em{\sc i\kern-.025em b}\kern-.08em
    T\kern-.1667em\lower.7ex\hbox{E}\kern-.125emX}}

\usepackage{algorithm}
\usepackage{amsthm}
\usepackage{bm, bbm}
\usepackage[capitalize]{cleveref}
\usepackage[inline]{enumitem}
\usepackage{mathtools}
\usepackage{nicefrac}
\usepackage{xspace}
\usepackage{csquotes}

\newtheorem{theorem}{Theorem}
\newtheorem{definition}{Definition}
\Crefname{definition}{Def.}{Defs.}%
\Crefname{algorithm}{Alg.}{Algs.}

\newcommand{\AlgNameLong}{Quadratic Continuous Quantum Optimization\xspace}
\newcommand{\AlgName}{QCQO\xspace}
\newcommand{\QUBO}{\textsc{Qubo}\xspace}

\newcommand{\beps}{\bm{\varepsilon}}
\newcommand{\ba}{\bm{a}}
\newcommand{\bb}{\bm{b}}
\newcommand{\bc}{\bm{c}}
\newcommand{\bA}{\bm{A}}
\newcommand{\bmu}{\bm{\mu}}
\newcommand{\bp}{\bm{p}}
\newcommand{\bP}{\bm{P}}
\newcommand{\bQ}{\bm{Q}}

\newcommand{\bR}{\bm{R}}
\newcommand{\bSigma}{\bm{\Sigma}}
\newcommand{\bu}{\bm{u}}
\newcommand{\bU}{\bm{U}}

\newcommand{\bw}{\bm{w}}
\newcommand{\bX}{\bm{X}}
\newcommand{\by}{\bm{y}}
\newcommand{\bz}{\bm{z}}
\newcommand{\bZ}{\bm{Z}}
\newcommand{\BB}{\mathbb{B}}
\newcommand{\BE}{\mathbb{E}}
\newcommand{\BN}{\mathbb{N}}

\newcommand{\BR}{\mathbb{R}}
\newcommand{\cL}{\mathcal{L}}
\newcommand{\cN}{\mathcal{N}}

\newcommand{\cR}{\mathcal{R}}
\newcommand{\cU}{\mathcal{U}}

\newcommand{\range}[1]{\lbrace 1,\dots,#1\rbrace}
\newcommand{\T}{\intercal}

\DeclareMathOperator{\diag}{diag}

\DeclareMathOperator{\MSE}{MSE}
\DeclareMathOperator{\vect}{vec}

\DeclarePairedDelimiter\ceil\lceil\rceil
\DeclarePairedDelimiter\norm\lVert\rVert
\DeclarePairedDelimiter\set\lbrace\rbrace

\begin{document}

\title{Quadratic Continuous Quantum Optimization
\thanks{This research has been funded by the Federal Ministry of Education and Research of Germany and the state of North Rhine-Westphalia as part of the Lamarr Institute for Machine Learning and Artificial Intelligence.}}

\author{\IEEEauthorblockN{Sascha M\"ucke}
\IEEEauthorblockA{\textit{Lamarr Institute} \\
\textit{TU Dortmund University}\\
Dortmund, Germany \\
sascha.muecke@tu-dortmund.de}
\and
\IEEEauthorblockN{Thore Gerlach}
\IEEEauthorblockA{\textit{Lamarr Institute} \\
\textit{University of Bonn}\\
Bonn, Germany\\
tgerlac1@uni-bonn.de}
\and
\IEEEauthorblockN{Nico Piatkowski}
\IEEEauthorblockA{\textit{Lamarr Institute} \\
\textit{Fraunhofer IAIS}\\
Sankt-Augustin, Germany \\
nico.piatkowski@iais.fraunhofer.de}}

\maketitle

\begin{abstract}
Quantum annealers can solve QUBO problems efficiently but struggle with continuous optimization tasks like regression due to their discrete nature.
We introduce Quadratic Continuous Quantum Optimization (QCQO), an anytime algorithm that approximates solutions to unconstrained quadratic programs via a sequence of QUBO instances.
Rather than encoding real variables as binary vectors, QCQO implicitly represents them using continuous QUBO weights and iteratively refines the solution by summing sampled vectors.
This allows flexible control over the number of binary variables and adapts well to hardware constraints.
We prove convergence properties, introduce a step size adaptation scheme, and validate the method on linear regression.
Experiments with simulated and real quantum annealers show that QCQO achieves accurate results with fewer qubits, though convergence slows on noisy hardware.
Our approach enables quantum annealing to address a wider class of continuous problems.
\end{abstract}

\begin{IEEEkeywords}
Quantum Annealing, QUBO, Continuous Optimization, Anytime Algorithms, Linear Regression.
\end{IEEEkeywords}

\section{Introduction}

Quantum Computing (QC) has emerged as a promising technology for accelerating algorithms, with well-known quantum subroutines such as Grover search~\cite{grover.1996a} or Shor's algorithm~\cite{shor.1997a} allowing, in theory, for asymptotic speedups using quantum circuits.
However, technical challenges such as quantum gate noise, decoherence and read-out errors prohibit large-scale application of quantum gate computing at the current point in time~\cite{preskill.2018a}.

Quantum Annealing (QA) is an alternative QC paradigm designed to solve one particular class of optimization problems~\cite{kadowaki.nishimori.1998a}.
Although it is, unlike quantum circuits, not a universal computation model, QA has become a valuable tool in combinatorial optimization, as it can solve various problems from satisfyability~\cite{kochenberger.etal.2005a} over graph problems~\cite{lucas.2014a} to machine learning~\cite{date.etal.2020a}.
Applications range from finance~\cite{hammer.shlifer.1971a} over logistics~\cite{neukart.etal.2017a,chai.etal.2023a} to image processing~\cite{piatkowski.etal.2022a}.
The central optimization problem class that quantum annealers solve is \QUBO.
\begin{definition}[\QUBO]
	Let $\bQ\in\BR^{n\times n}$ for some fixed $n\in\BN$.
	The \emph{Quadratic Unconstrained Binary Optimization} problem is to find a binary vector $\bz^*\in\BB^n$ that minimizes the \emph{energy function} \begin{equation*}
		E_{\bQ}(\bz)=\bz^\T\bQ\bz=\sum_{i=1}^n\sum_{j=1}^nQ_{ij}z_iz_j.
	\end{equation*}
\end{definition}
We typically assume $\bQ$ to be either symmetrical or an upper triangle matrix, such that it defines the energy function uniquely.
Every non-symmetrical \QUBO instance can be made symmetrical by applying $\bQ\gets (\bQ+\bQ^\T)/2$.
\QUBO is generally strongly \textsf{NP}-hard~\cite{cela.punnen.2022a}.
Besides QA, a range of solution or approximation strategies for \QUBO exist, both heuristic and exact~\cite{punnen.etal.2022a}.

A common approach in literature is to cast optimization problems arising in applications into \QUBO form, solve them using a quantum annealer and interpreting the resulting minimizing binary vector as the solution to the original problem.
While this approach has proven very fruitful for a large number of problems~\cite{lucas.2014a}, a strong limitation of \QUBO is that it is fundamentally combinatorial, having a discrete and finite domain.
This makes it hard to apply it naturally to problems with a continuous solution space, such as regression or artificial neural network training, where the variables are real-valued weights.
Nonetheless, \QUBO formulations of both of these problems have been devised, using rough approximations.
For instance, Date et al. \cite{date.etal.2021a} uses a \emph{precision vector} $\bp\in\BR^p$ to interpret a binary vector $\bz\in\BB^{np}$ as digits of a base, such that the real-valued solution is given by $\bw=\bP\bz\in\BR^n$ with $\bP=\bm{1}_n\otimes\bp^\T$, where $\bm{1}_n$ denotes the $n\times n$ identity matrix, and $\otimes$ is the Kronecker product.
Naturally, the elements of $\bw$ are still discrete, taking at most $2^p$ distinct values.
Moreover, the number of binary variables is $\mathcal{O}(np)$, meaning that doubling the precision implies doubling the number of variables, which is a prohibitive factor for contemporary QA hardware.

In this article we present a novel approach to solving optimization problems with continuous domain on quantum annealers.
To this end, we generalize the idea of precision vectors by representing numbers as sums of vectors we repeatedly sample from a random distribution, refining the result with every \QUBO instance we solve.
Our contribution has a number of advantages:
Firstly, we eliminate the necessity to explicitly represent real numbers as binary vectors by representing them implicitly through the \QUBO weights, which are already continuous.
Secondly, the number of vectors, i.e., the number of binary variables to optimize per \QUBO instance, is a hyperparameter that can be chosen freely according to the hardware capacities, making this approach very flexible.
Lastly, the resulting strategy is an anytime algorithm, as each consecutive \QUBO solution is a refinement of the previous one.
This allows for easy adaption to quantum annealers with a limited budget.

This paper is structured as follows:
In \cref{sec:qcqo} we describe our proposed continuous optimization procedure.
After a discussion of its theoretical properties in \cref{sec:theory} we put it into practice in \cref{sec:lr}, where we use our method to perform regression as an exemplary optimization task.
We improve the performance by introducing a step size scheduling scheme in \cref{sec:stepsize}.
Section \cref{sec:conclusion} discusses the implications of possible future developments of this method and concludes this article.

\section{Quadratic Continuous Quantum Optimization}
\label{sec:qcqo}

Given an optimization problem with a quadratic loss function $\cL:\BR^d\rightarrow\BR$ of the form \begin{align}\label{eq:qp}
	\cL(\bw) &= \bw^\T\bA\bw+\ba^\T\bw+c,
\end{align}
we want to find a vector of \emph{weights} $\bw^*\in\BR^d$ that minimizes $\cL$.
Without loss of generality we can assume that $\bA$ is symmetrical (if it is not, we can set $\bA\gets(\bA+\bA^\T)/2$).
To further simplify, we can ignore $c$ for the purpose of minimization.
This problem is equivalent to quadratic programming without any linear constraints~\cite{nocedal.wright.2006a}.
If $\bA$ is indefinite, the problem is \textsf{NP}-hard.

Now assume that $\bw\in\BR^d$ is an initial guess or an intermediate result which we want to refine, e.g., $\bw=\bm{0}$.
Given a matrix $\bR\in\BR^{n\times d}$, we can define a \QUBO instance that finds \begin{equation}
	\bz^*=\underset{\bz\in\BB^n}{\arg\min}~\cL(\bw+\bR^\T\bz),
\end{equation}
where $\bR^\T\bz$ can be interpreted as the sum of a subset of rows of $\bR$.
We find that this \emph{update} to the solution $\bw$ has itself a quadratic form:
\begin{align*}
	\cL(\bw+\bR^\T\bz)\hspace{-5em}\\
	&=(\bw+\bR^\T\bz)^\T\bA(\bw+\bR^\T\bz)+\ba^\T(\bw+\bR^\T\bz)\\
	&= \bw^\T\bA\bw+\bw^\T\bA\bR^\T\bz+\bz^\T\bR\bA\bw+\bz^\T\bR\bA\bR^\T\bz\\
	&\hspace{2em}+\ba^\T\bw+\ba^\T\bR^\T\bz\\
	&= \bz^\T\bR\bA\bR^\T\bz+(2\bw^\T\bA+\ba^\T)\bR^\T\bz + \mathrm{const}.
\end{align*}

As we can see, this is an $n$-variable \QUBO instance with weight matrix \begin{equation}\label{eq:qcqo-qubo}
	\bQ(\bw,\bR)=\bR\bA\bR^\T+\diag[\bR(2\bA\bw+\ba)],
\end{equation}%
where $\diag[\bm{x}]=\bX$ constructs a diagonal matrix such that $X_{ii}=x_i\,\forall i$ and $X_{ij}=0$ if $i\neq j$.

Solving this problem yields an indicator vector $\bz^*\in\BB^n$ that lets us perform a weight update $\bw'\gets\bw+\bR^\T\bz^*$, such that $\cL(\bw')\leq\cL(\bw)$.
The number of variables of the \QUBO problem corresponds to the number of rows $n$ of $\bR$, which is a hyperparameter:
A large $n$ allows for more row sums, and thus a higher probability of finding a weight update that improves the loss value, at the cost of a larger \QUBO size.
By repeating this weight update and sampling $\bR$ from some probability distribution $\cR$ over $\BR^{n\times d}$ every time, we can iteratively refine the solution.
This idea is formalized in \cref{alg:qcqo}.

\begin{algorithm}[H]
	\caption{\AlgNameLong}
	\label{alg:qcqo}
	\begin{algorithmic}[1]
		\renewcommand{\algorithmicrequire}{\textbf{Input:}}
		\renewcommand{\algorithmicensure}{\textbf{Output:}}
		\REQUIRE Quadratic problem defined through $\bA\in\BR^{d\times d}$ and $\ba\in\BR^d$ as in \cref{eq:qp}; probability distribution $\cR$ over $\BR^{n\times d}$; initial guess $\bw_0\in\BR^d$; budget or stopping criterion
		\ENSURE Improved solution $\bw^\circ$
		\STATE $t\gets 0$
		\WHILE {stopping criterion unfulfilled or budget remains}
			\STATE Sample $\bR_t\sim\cR$
			\STATE Compute $\bQ(\bw_t,\bR_t)$ according to \cref{eq:qcqo-qubo}
			\STATE $\bz^*_t\gets\arg\min_{\bz}E_{\bQ}(\bz)$ \COMMENT{use quantum annealer}
			\STATE $\bw_{t+1}\gets\bw_t+\bR_t^\T\bz^*_t$
			\STATE $t\gets t+1$
		\ENDWHILE
		\RETURN $\bw_t$
	\end{algorithmic}
\end{algorithm}

\subsection{Sampling the Matrix $\bR$}

The choice of distribution $\cR$ is of central importance for the behavior of \AlgName, as it controls the distribution of possible update steps $\bR^\T\bz$.

While it is in principle possible to use an arbitrary matrix-valued distribution, we can sample the rows $\bR_1,\dots,\bR_n$ independently from some multivariate distribution, or even the individual elements $R_{ij}$ from some univariate distribution.
If we sample row-wise and i.i.d. from a multivariate normal distribution with mean $\bm\mu$ and covariance matrix $\bm\Sigma$ (denoted by $\cN(\bmu,\bSigma)$), it is easy to see that every point in $\bR^d$ can be reached with a single update step, as follows from the fact that $\cN(\bmu,\bSigma)$ has full support on $\BR^d$.
However, we can show that the distribution of the update step $\bU=\bR^\T\bz$ is closely related to the distribution of each row.

\begin{theorem}\label{thm:normal}
	For an arbitrary but fixed $n\in\BN$, if $\bR_i$ is sampled i.i.e. from $\cN(\bmu,\bSigma)$ for each row $i\in\range{n}$, the expected update step $\bU$ in each iteration of \cref{alg:qcqo} follows the distribution $\cN(\nicefrac{n}{2}\bmu,\nicefrac{n}{4}\bSigma)$.
\end{theorem}%
\begin{proof}
Consider the random variable $\bP=\vect(\bR)$, which follows a multivariate normal distribution $\cN(\bm{1}_n\otimes\bmu,\bm{I}_n\otimes\bSigma)$, where $\bm{1}_n=[1,1,\dots,1]^\T\in\BR^n$, $\bm{I}_n$ is the $n\times n$ identity matrix, and $\otimes$ denotes the Kronecker product.
Assuming a uniform distribution of binary vectors $\bz\in\BB^n$, the update step is itself a random variable $\bU$ defined as \begin{align*}
	\bU &= \frac{1}{2^n}\sum_{\bz} \underbrace{(\bz^\T\otimes\bm{I}_d)\bP}_{\equiv \,\bR^\T\bz} \\
	&= \biggl(\biggl(\frac{1}{2^n}\sum_{\bz}\bz^\T\biggr)\otimes\bm{I}_d\biggr)\bP \\
	&= \underbrace{\nicefrac{1}{2}(\bm{1}^{\T}_n\otimes\bm{I}_d)}_{=:\bar{\bZ}}\bP.
\end{align*}
As this is a linear transformation on $\bP$, we can exploit the fact that $\bX\sim\cN(\bmu,\bSigma)\Rightarrow \bA\bX+\bm{b}\sim\cN(\bA\bmu+\bm{b},\bA\bSigma\bA^\T)$ and find that $\bU\sim\cN(\bmu_{\cU},\bSigma_{\cU})$ with \begin{align*}
	\bmu_{\cU}&=\bar{\bZ}(\bm{1}_n\otimes\bmu)\\
	&=\nicefrac{1}{2}(\bm{1}^{\T}_n\otimes\bm{I}_d)(\bm{1}_n\otimes\bmu)\\
	&=\nicefrac{1}{2}(\bm{1}^{\T}_n\bm{1}_n)\otimes(\bm{I}_d\bmu)\\
	&=\nicefrac{n}{2}\bmu,\\[\jot]
	\bSigma_{\cU}&=\bar{\bZ}(\bm{I}_n\otimes\bSigma)\bar{\bZ}^\T\\
	&=\nicefrac{1}{4}(\bm{1}^\T_n\otimes\bm{I}_d)(\bm{I}_n\otimes\bSigma)(\bm{1}_n\otimes\bm{I}_d)\\
	&=\nicefrac{1}{4}(\bm{1}^\T_n\bm{I}_n\bm{1}_n)\otimes(\bm{I}_d\bSigma\bm{I}_d)\\
	&=\nicefrac{n}{4}\bSigma.
\end{align*}
\end{proof}
This implies that the update steps follow themselves a multivariate normal distribution across all iterations.

\section{Theoretical Analysis}
\label{sec:theory}
We analyze the case when $\cL(\bw)=\bw^\T\bA\bw+\ba^\T\bw$ is convex. 

\begin{theorem}\label{thm:normal2}
	The sequence of iterates $\bw_t$ generated by \cref{alg:qcqo} satisfies
	\begin{align*}
	\cL(\bw_{t+1})-\cL(\bw^*) &\leq\frac{\norm{\bA}_2\norm{\bw_0-\bw^*}_2}{2t}+\frac{1}{t} \sum_{t'=1}^t\beps_t^\T\bR_t^\T\bz^*_t.
	\end{align*}
\end{theorem}%
\begin{proof}
	For any $\bw,\bw'\in\mathbb{R}^d$, we have \begin{equation*}
	\norm{\nabla \cL(\bw)-\nabla \cL(\bw')}_2=\norm{\bA(\bw-\bw')}_2\leq\norm{\bA}_2\,\norm{\bw-\bw'}_2.
	\end{equation*}
	Hence, $\nabla\cL$ is Lipschitz continuous with constant $L=\norm{\bA}_2$. 
	Thus, we have the quadratic upper bound~\cite{nesterov.2004a}
	\begin{equation*}
	\cL(\bb)\leq\cL(\bc)+\nabla\cL(\bc)^\T\bb-\bc+\nicefrac{L}{2}\norm{\bb-\bc}_2^2.
	\end{equation*}
	Setting $\bb=\bw_{t+1}$ and $\bc=\bw_t$, we arrive at
	\begin{equation*}
	\cL(\bw_{t+1})\leq\cL(\bw_t)+(\bA\bw_t+\ba)^\T\bR_t^\T\bz^*_t+\nicefrac{1}{2}\norm{\bA}_2\norm{\bR_t^\T\bz^*_t}_2^2.
	\end{equation*}
	Moreover, \begin{equation*}
	\cL(\bw_{t})\leq\cL(\bw^*)+\nicefrac{1}{2}\norm{\bA}_2 \norm{\bw_t-\bw^*}_2^2,\end{equation*}
	and thus \begin{align*}
	\cL(\bw_{t+1})-\cL(\bw^*)\leq\hspace{-2cm}&\\
		&(\bA\bw_t+\ba)^\T\bR_t^\T\bz^*_t+\nicefrac{1}{2}\norm{\bA}_2\norm{\bR_t^\T\bz^*_t}_2^2\\
		&+\nicefrac{1}{2}\norm{\bA}_2\norm{\bw_t-\bw^*}_2^2.
	\end{align*}
	Now, let $\bR_t^\T\bz^*_t = -\norm{\bA}_2(bA\bw_t+\ba)+\beps_t$	for all $t$.
	We find \begin{align*}
	\cL(\bw_t)\leq\hspace{-5mm}&\\
		&\cL(\bw^*)+\nicefrac{1}{2}\norm{\bA}_2(\norm{\bw_t-\bw^*}_2^2-\norm{\bw_{t+1}-\bw^*}_2^2)\\
		&+\beps_t^\T\bR_t^\T\bz^*_t.
	\end{align*}
	Any function value $\cL(\bw_{t'})$ with $t'<t$ cannot be closer to $\cL(\bw^*)$ than the function value of the $t$-th iterate $\cL(\bw_t)$, and so is their average.
	Finally, we get \begin{align*}
	\cL(\bw_{t+1})-\cL(\bw^*)\hspace{-2cm}&\\
		&\leq\frac{1}{t}\sum_{t'=1}^t \cL(\bw_{t'})-\cL(\bw^*)\\
		&\leq\frac{\norm{\bA}_2\norm{\bw_0-\bw^*}_2}{2t}+\frac{1}{t} \sum_{t'=1}^t\beps_t^\T\bR_t^\T\bz^*_t.
	\end{align*}
\end{proof}
The result implies that \cref{alg:qcqo} can enjoy (at least) the convergence speed of gradient descent (with fixed stepsize) when the distribution $\cR$ is chosen appropriately.
Note, however, that this result is very pessimistic, since gradient descent must perform a step is a neighborhood of $\bw_t$ while \cref{alg:qcqo} may perform arbitrary large jumps.
We postpone the analysis of such cases to future work. 

\subsection{Number of Qubits}
Approaches comparable to \AlgName use precision vectors to represent continuous values~\cite{date.etal.2020a,date.etal.2021a}.
For encoding a value in the unit interval $[0,1]$ using $k$ (qu)bits, we can construct a vector $\bp = (2^{i-1}/2^k-1)^{\T}_{i=1,\dots,k}$ so that $\mathcal{X}_{\bp}=\set{\bp^\T\bz:~\bz\in\BB^k}$ samples the interval evenly in steps of size $2^k-1$, with $\bp^\T\bm{0}=0$ and $\bp^\T\bm{1}=1$.
To represent any number within a neighborhood of size $\epsilon>0$ we need at least $k\geq \ceil{\log_2(1/(2\epsilon)+1)}$, or $\Theta(d\log(1/\epsilon))$ (qu)bits.

In contrast, \AlgName makes the continuous values implicit and instead uses (qu)bits to encode the inclusion/exclusion of fixed summands in a sum.
Consequently our method can encode any arbitrary vector $\bm{x}\in\BR^d$ using just a single bit $z_i$, by simply choosing $\bR_i=\bm{x}$, i.e., we need $\Theta(1)$ binary variables in general.
This allows us to choose the number of (qu)bits flexibly according to hardware limitations, requiring fewer qubits than explicit encodings.

\section{Application: Linear Regression}
\label{sec:lr}

Given a data matrix $\bX\in\BR^{N\times d}$ consisting of $N$ rows of $d$-dimensional data points and a vector $\by\in\BR^N$ of target values, we want to find a weight vector $\bw\in\BR^d$ and a bias term $b\in\BR$ such that $\bX\bw+b$ is an approximation of $\by$.
The approximation quality is quantified by the Mean Squared Error (MSE) defined as \begin{equation}
	\MSE(\bw,b;\bX,\by) = \nicefrac{1}{N}\norm{\bX\bw+b\bm{1}-\by}_2^2.
\end{equation}
For convenience, we fuse the bias term with the weight vector by assuming $w_d=b$ and $X_{i,d}=1~\forall i$.
This lets us write the MSE in matrix-vector form as \begin{align}
	\MSE(\bw;\bX,\by)\hspace{-6em}\\
	&= \nicefrac{1}{N}(\bX\bw-\by)^\T(\bX\bw-\by) \nonumber\\
	&= \nicefrac{1}{N}(\bw^\T\bX^\T\bX\bw-2\by^\T\bX\bw+\by^\T\by)\nonumber\\
	&= \bw^\T(\underbrace{\nicefrac{1}{N}\bX^\T\bX}_{=:\bA})\bw+(\underbrace{\nicefrac{-2}{N}\bX^\T\by}_{=:\ba})^\T\bw+\mathrm{const}.\label{eq:lr}
\end{align}

Applying \cref{eq:qcqo-qubo} by replacing $\bA$ and $\ba$ with above terms we obtain a \QUBO instance with weight matrix \begin{align*}
	\bQ_{\mathrm{LR}}(\bw,\bR;\bX,\by)\hspace{-6em}\\
	&=\nicefrac{1}{N}\left(\bR\bX^\T\bX\bR^\T+2\diag[\bR\bX^\T(\bX\bw-\by)]\right)\!.
\end{align*}

\subsection{Experimental Evaluation}

To test the convergence behavior of \AlgName, we generate a synthetic regression data set with $d=16$ and $N=10^5$ by \begin{enumerate*}[label=(\roman*)]\item sampling $\bw\in\cN(\bm{0}_{d},\bm{I}_{d})$, \item setting $\bw\gets 100\bw/\norm{\bw}_2$, such that the weight vector has a fixed Euclidean distance of $100$ from the origin, \item sampling $\bX\in\BR^{N\times d}$ with $\bX_i\sim\cN(\bm{0}_{d}, d\bm{I}_d)$, \item setting $X_{i,d}=1\,\forall i$, and \item computing $\by=\bX\bw$\end{enumerate*}.
Following this approach we obtain a regression data set with ground-truth weight vector $\bw^*$, which we try to recover using \AlgName.

To this end, we initialize \cref{alg:qcqo} with $\bA$ and $\ba$ according to \cref{eq:lr} and choose $\bm{0}_d$ as our initial guess for every run.
The rows of $\bR$ are sampled i.i.e. from a normal distribution $\cN(\bm{0}_d,\nicefrac{4\sigma}{n}\bm{I}_d)$, which by \cref{thm:normal} results in update steps distributed as $\cN(\bm{0}_d,\sigma\bm{I}_d)$ in expectation over uniform $\bz$.
We choose all combinations of $(n,\sigma)\in\set{8, 16, 24}\times\set{0.1, 1.0}$.
Due to the relatively small values of $n$, we use a brute-force \QUBO solver.

\begin{figure}
	\centering
	\includegraphics[width=\columnwidth]{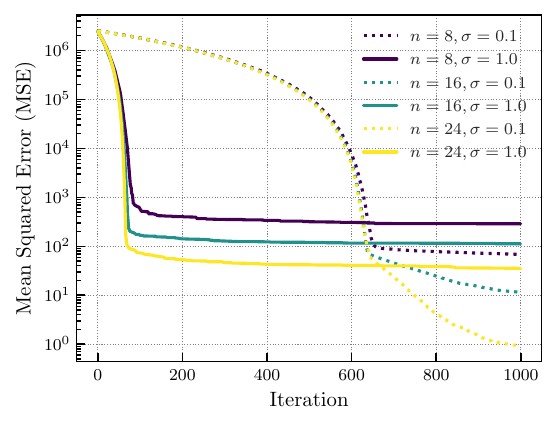}
	\caption{Convergence behavior of \AlgName on Linear Regression tasks with different choices of $n$ (number of rows of update matrix $\bR$) and $\sigma$ (standard deviation used for sampling the rows of $\bR$); MSE is averaged over 10 runs.}
	\label{fig:lr}
\end{figure}

\begin{figure}
	\centering
	\includegraphics[width=\columnwidth]{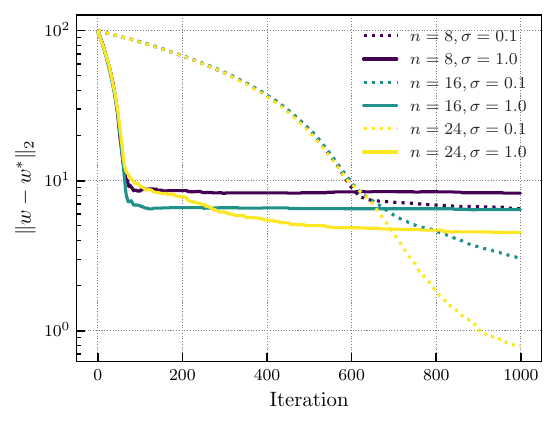}
	\caption{Same experiment as \cref{fig:lr}, but showing the Euclidean distance between current solution $\bw$ of each iteration and the globally optimal solution $\bw^*$, averaged over 10 runs.}
	\label{fig:lr_w}
\end{figure}

The results are shown in \cref{fig:lr}, from which we make two key observations:
Firstly, we observe that convergence is faster with higher $n$.
This is likely due to the larger number of sums that can be formed from the rows of $\bR$ in each iterations.
Increasing $n$ leads to an exponentially larger set of possible steps in $\BR^d$, increasing the probability to find better solutions.

Secondly, a larger $\sigma$ leads to a steeper initial descent followed by a shallower descent after an MSE of around $10^3$.
We assume this change in behavior occurs when the distance between the current solution and the optimal solution becomes smaller than a standard deviation of $\bU$.
This is analogous to the behavior of Gradient Descent around the optimum when the step size is too large and the current solution starts oscillating around the true optimum~\cite{nesterov.2004a}. 
However, the key difference is that \AlgName cannot move to a solution with higher loss value.

The second observation suggests that \AlgName would profit from a gradually decreasing $\sigma$, following some schedule similar to a learning rate schedule in evolutionary strategies or stochastic gradient descent learning~\cite{bottou.etal.2018a,glasmachers.2020a}.

\Cref{fig:lr_w} shows the same experiment as \cref{fig:lr}, but reporting the Euclidean distance between the current solution $\bw_t$ of each iteration and the globally optimal vector $\bw^*$.
The curves are very similar to the MSE, showing a steep approach up to around 80 iterations, after which the distance stagnates, with larger $n$ leading to a smaller distance.
Again, a lower $\sigma$ leads to a slower initial approach, and a much smaller final distance after 1000 iterations.

\section{Step Size Scheduling}
\label{sec:stepsize}

To introduce an adaptive step size, we modify \AlgName slightly by \begin{enumerate*}[label=(\roman*)]\item assuming that the distribution $\cR$ is a row-wise isotropic multivariate normal distribution $\cN(\bm{0}_d, \nicefrac{4\sigma}{n}\bm{I}_d)$, and \item updating $\sigma$ in every iteration based on a window of past update steps, for which \item we introduce a window size hyperparameter $T\in\BN$\end{enumerate*}.

As explained before, a standard deviation of $\nicefrac{4\sigma}{n}$ for the rows of $\bR$ leads to an expected standard deviation of $\sigma$ on the update step $\bU=\BE_{\bz}[\bR^\T\bz]=\nicefrac{1}{2}\bm{1}_n^\T\bR$.
By computing the average size of the past $T$ update steps, we can deduce the approximate step size that should be taken that has lead to the largest improvement over the past $T$ iterations.

\begin{algorithm}[H]
	\caption{\AlgNameLong with window-based step size adaption}
	\label{alg:qcqo_stepsize}
	\begin{algorithmic}[1]
		\renewcommand{\algorithmicrequire}{\textbf{Input:}}
		\renewcommand{\algorithmicensure}{\textbf{Output:}}
		\REQUIRE Quadratic problem defined through $\bA\in\BR^{d\times d}$ and $\ba\in\BR^d$ as in \cref{eq:qp}; number of rows $n\in\BN$ (\QUBO size); window size $T\in\BN$; initial guess $\bw_0\in\BR^d$; budget or stopping criterion
		\ENSURE Improved solution $\bw^\circ$
		\STATE $t\gets 0$
		\WHILE {stopping criterion unfulfilled or budget remains}
		\IF {$t>T$}
		\STATE $\sigma_t\gets \nicefrac{1}{T}\sum_{\tau=1}^T\norm{\bu_{t-\tau}}_2$
		\ELSE
		\STATE $\sigma_t\gets 1$
		\ENDIF
		\STATE Sample every row of $\bR_t$ i.i.d. from $\cN(\bm{0}_d,\sigma_t\bm{I}_d)$
		\STATE Compute $\bQ(\bw_t,\bR_t)$ according to \cref{eq:qcqo-qubo}
		\STATE $\bz^*_t\gets\arg\min_{\bz}E_{\bQ}(\bz)$ \COMMENT{use quantum annealer}
		\STATE $\bu_t\gets\bR_t^\T\bz^*_t$
		\STATE $\bw_{t+1}\gets\bw_t+\bu_t$
		\STATE $t\gets t+1$
		\ENDWHILE
		\RETURN $\bw_t$
	\end{algorithmic}
\end{algorithm}

To this end, let $\bu_{t}$ denote the update step in iteration $t$.
We compute \begin{equation}\label{eq:stepsize}
	h_t = \nicefrac{1}{T}\sum_{\tau=1}^T\norm{\bu_{t-\tau}}_2 = \nicefrac{1}{T}\sum_{\tau=1}^T\sqrt{\sum_{i=1}^du^2_{t-\tau,i}}
\end{equation}%
as an online estimator of the square root of the 2\textsuperscript{nd} moment of the update step, which we use as a new value of sigma, setting $\sigma=h_t$.
See \cref{alg:qcqo_stepsize} for an overview of the procedure.

\begin{figure}
	\centering
	\includegraphics[width=\columnwidth]{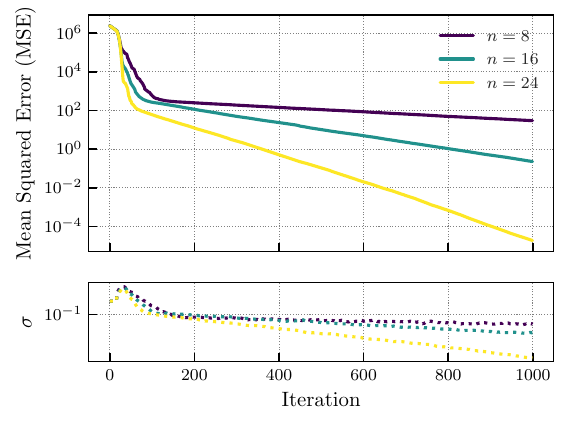}
	\caption{Convergence behavior of \AlgName on Linear Regression tasks with different choices of $n$ (number of rows of update matrix $\bR$) and $\sigma$ (standard deviation used for sampling the rows of $\bR$), using the step size update rule described by \cref{eq:stepsize}; MSE is averaged over 10 runs.}
	\label{fig:lr_stepsize}
\end{figure}

\begin{figure}
	\centering
	\includegraphics[width=\columnwidth]{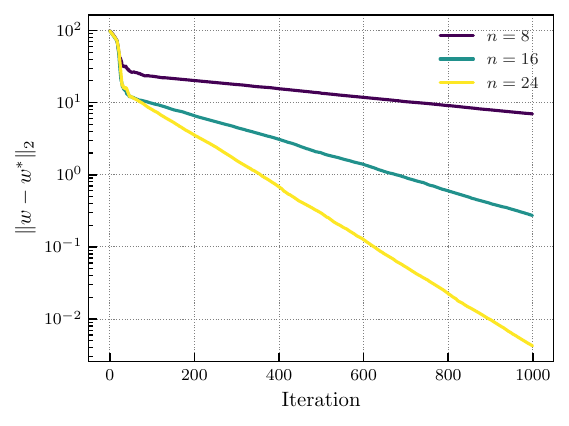}
	\caption{Same experiment as \cref{fig:lr_stepsize}, but showing the Euclidean distance between current solution $\bw_t$ of each iteration and the globally optimal solution $\bw^*$, averaged over 10 runs.}
	\label{fig:lr_w_stepsize}
\end{figure}

The results are shown in \cref{fig:lr_stepsize}, again for multiple $n$ on the same data set as before.
As we expected, initially the value of $\sigma$ increases to around $10$ before decreasing steadily, leading to a much lower MSE value after the budget of 1000 iterations was depleted.
Like before, a larger $n$ is beneficial for faster convergence at the cost of a larger \QUBO search space.

Again, the Euclidean distance between $\bw_t$ in each iteration and the globally optimal $\bw^*$ is shown in \cref{fig:lr_w_stepsize}, and again it is very similar to \cref{fig:lr_stepsize}, leading to a much smaller final distance after 1000 iterations than without step size adaption.

\subsection{Evaluation on Quantum Annealer}

For small $n<30$ we were able to compute the ground-truth optima of the \QUBO instances in \cref{alg:qcqo,alg:qcqo_stepsize} by brute force, but this quickly becomes infeasible for larger $n\gg 30$.
To investigate the performance of \AlgName using imperfect \QUBO solvers, we run it using a D-Wave quantum annealer.
Specifically, we use the Advantage 2 system 1.1, which comprises over 5,000 qubits and up to around 35,000 couplings using the Pegasus topology~\cite{d-wavesystems.2021a}, allowing for dense \QUBO weight matrices of up to $n=180$.
We solve each \QUBO instance with an annealing time of 20 ns and 100 reads (``shots''): the annealing process is repeated 100 times per instance and the lowest-energy solution returned.
Due to persisting challenges of real quantum hardware, this solution may not be optimal~\cite{d-wave.2023a}.
However, \AlgName should still work, as long as the solution returned by the solver does not \emph{increase} the loss value.
This is highly unlikely, as $\bm{0}_n$ is always a valid solution leaving $\bw_t$ unchanged.

\begin{figure}
	\centering
	\includegraphics[width=\columnwidth]{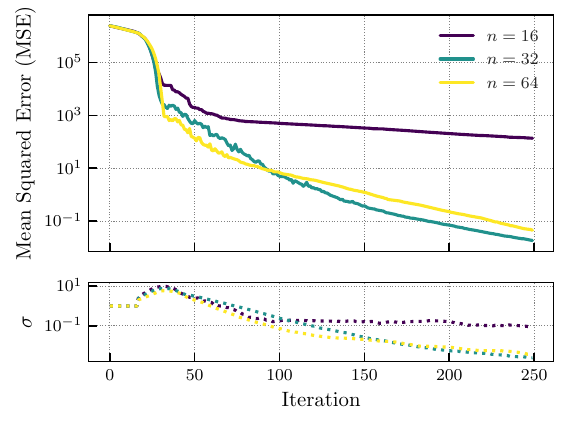}
	\caption{Same experiment as \cref{fig:lr_stepsize}, but run on a D-Wave quantum annealer, averaged over 10 runs.}
	\label{fig:lr_dwave}
\end{figure}

We repeat the same experiment described in \cref{sec:stepsize}, but replace the brute-force \QUBO solver with calls to the D-Wave device.
Additionally, we run the algorithm with $n\in\set{16,32,64}$, as the QA device allows us to solve larger \QUBO instances.
The resulting MSE over time is shown in \cref{fig:lr_dwave}:
We clearly observe a very similar behavior as in \cref{fig:lr_stepsize}, although the decrease in MSE is less steep, leading to a final MSE of only around $10^{-1}$ after 1000 iterations for $n=64$.
Interestingly, the final MSE is better for $n=32$ than for $n=64$, whereas using the brute-force solver a higher $n$ always lead to a lower loss value.
Presumably, this is due to imperfect solutions returned by the QA device, leading to higher-quality solutions for $n=32$ than for $n=64$, which seems to be more important for steady convergence than the additional update step directions introduced by more rows of $\bR$.

\section{Conclusion}
\label{sec:conclusion}

In this work we have described the \AlgNameLong (\AlgName) algorithm, approximating solutions of unconstrained quadratic programming problems using a series of \QUBO problems, which can be solved by QA.
We demonstrated that the algorithm works and finds good solutions to an LR problem.
In addition, we devised a step size adaption scheme that demonstrably improves convergence speed and solution quality.
Lastly, we showed that \AlgName also works on real quantum hardware by running it on a D-Wave Advantage 2, though the reduced solution quality leads to slower convergence.

In the scope of this work, we only used LR as an exemplary optimization problem, which has a closed-form solution, however, the class of problems that can be approached by \AlgName is much larger and comprises harder problems.
In particular, it can be applied naturally to all unconstrained quadratic programming problems with indefinite parameter matrix $\bA$, which arise in diverse settings, including relaxations of combinatorial problems such as Max-Cut and \QUBO, as subproblems in robust and nonconvex optimization~\cite{burkhard.etal.1998a}, or as intermediate formulations in finance and signal processing; our method thus offers potential utility as a heuristic in these domains.
In future work, we will explore the efficiency of our approach on other problems, particularly such with non-convex loss functions.

So far, the matrix $\bR$ is always sampled from some distribution.
While this approach has desirable theoretical properties, the question arises if some fixed choice of $\bR$ leads to the best possible coverage of the hyperball around $\bw_t$ across all $\bz$, or if we can adapt $\bR$ (or its distribution parameters) based on the previous update steps, similarly to covariance matrix adaption~\cite{hansen.etal.2003a}.

Our step size adaption scheme (\cref{sec:stepsize}), while effective in our experiments, is largely heuristic.
An in-depth analysis of different schemes and how they impact the convergence properties and efficiency of \AlgName is left for future work.


\end{document}